\documentclass[twocolumn, twoside]{IEEEtran}

\usepackage[utf8]{inputenc}
\usepackage[T1]{fontenc}
\usepackage{url}
\usepackage{ifthen}
\usepackage{cite}
\usepackage{booktabs}
\usepackage[cmex10]{amsmath} 

\usepackage[pdftex]{graphicx}
\usepackage[cmex10]{amsmath}
\usepackage{amsfonts}
\usepackage{amssymb}
\usepackage{amsxtra}
\usepackage{latexsym}
\usepackage{subfigure}
\usepackage{cite}
\usepackage{xcolor, soul}
\usepackage{url}
\usepackage{enumitem}
\usepackage{bbm}
\usepackage{multicol, lipsum}
\usepackage{amsthm}
\usepackage{multirow}
\usepackage[linesnumbered,ruled,vlined]{algorithm2e}
\usepackage{stfloats}
\usepackage{kotex}
\usepackage{mathtools, stmaryrd}
\usepackage{xparse} \DeclarePairedDelimiterX{\Iintv}[1]{\llbracket}{\rrbracket}{\iintvargs{#1}}
\NewDocumentCommand{\iintvargs}{>{\SplitArgument{1}{,}}m}
{\iintvargsaux#1} %
\NewDocumentCommand{\iintvargsaux}{mm} {#1\mkern1.5mu..\mkern1.5mu#2}

\SetCommentSty{mycommfont}

\SetKwInput{KwInput}{Input}                
\SetKwInput{KwOutput}{Output}              
\SetKwInput{KwInit}{Initialization}              
\DeclareMathOperator*{\argmin}{argmin}
\DeclareMathOperator*{\argmax}{argmax}
\DeclareMathOperator{\sign}{sign}
\DeclareMathOperator{\rank}{rank}

\renewcommand{\thefootnote}{\fnsymbol{footnote}}
\theoremstyle{definition}
\newtheorem{definition}{Definition}
\theoremstyle{lemma}
\newtheorem{lemma}{Lemma}
\theoremstyle{theorem}
\newtheorem{theorem}{Theorem}
\theoremstyle{corollary}
\newtheorem{corollary}{Corollary}

\ifCLASSOPTIONonecolumn
  \renewcommand{\baselinestretch}{1.57}
  \newcommand{\figwidth}{0.6\columnwidth}

\else
  \renewcommand{\baselinestretch}{0.99}
  \newcommand{\figwidth}{.94\columnwidth}

\fi

\begin{document}
\title{Design of Polar Codes with Puncturing and Extending}


\author{%
 \IEEEauthorblockN{Seokju Han~\IEEEmembership{Member,~IEEE}, Inayat~Ali~\IEEEmembership{Member,~IEEE}, Bonghoe~Kim, and Jeongseok Ha~\IEEEmembership{Senior Member,~IEEE}}\\%
}

\maketitle

\begin{abstract}
This work proposes a novel polar code design scheme using puncturing and extending aiming to improve the successive cancellation list decoding (SCLD) performance. The proposed code design is conducted with Reed-Muller (RM) codes called \emph{base} codes. For a base code, puncturing is first performed to gain a degree of freedom, called \emph{transmission holes}. To this end, we develop a puncturing scheme minimizing the loss in the minimum distance of the base code. Then, we propose a novel extending scheme to produce additional bits which are transmitted through the transmission holes without rate loss. The extending scheme repeats some coded bits and/or intermediate bits induced during the encoding of base code. The repeated bits are decided in such a way to maximize the SCLD performance by utilizing reinforcement learning (RL). Moreover, we introduce a novel idea to reduce the learning space of RL, which greatly expedites the design of extended bits. While the proposed design scheme assumes RM codes as base codes, we show that it can be applied to various types of polar code. Simulation results show that the polar codes designed with the proposed scheme have notably improved SCLD performance compared to existing polar code designs.
\end{abstract}

\ifCLASSOPTIONonecolumn
  \clearpage
  \pagenumbering{arabic}
\fi

\section{Introduction}\label{Sec:Intro}
\IEEEPARstart{P}{olar} codes are shown to achieve the channel capacity of any symmetric binary-input memoryless channel with a successive cancellation (SC) decoder \cite{arikan2009channel}. By using the polarization kernel, the channel combining and splitting operations in the polar codes transform $N$-physical channels into $N$-synthesized channels that are either highly reliable or noisy \cite{arikan2009channel}. For the highly reliable channel indices, called \emph{information set} $\mathcal I$, the user message bits, also called \emph{information bits}, are transmitted over it, whereas, the unreliable channel indices, called \emph{frozen set} $\mathcal F$, are set to known bit values, called \emph{frozen bits} (zeros in this work), at both transmitter and receiver. As the code length increases, the reliabilities go to infinity for bit channel indices in $\mathcal I$ and become completely noisy for channel indices in $\mathcal F$.


The conventional polar codes were designed based on the selection of information set $\mathcal I$ to minimize the error rate performance of SC decoding (SCD) for target code rates \cite{arikan2009channel, mori2009performance, trifonov12efficient, He2017beta}. In \cite{arikan2009channel}, the information set $\mathcal I$ is selected based on the Bhattacharyya parameter. Whereas, the density evolution (DE) technique is utilized in \cite{mori2009performance} to estimate the SCD performance and optimize the construction of the polar codes. To reduce the computational complexity of DE, a Gaussian approximation-based construction scheme is proposed in \cite{trifonov12efficient}, and the designed codes are shown to have comparable performance with the code designed with DE \cite{mori2009performance}. In \cite{schurch2016partial}, the existence of a universal partial order of bit-channel reliabilities was shown for the polar codes. Based on the universal partial order, a method to construct polar codes with reduced design complexity is proposed in \cite{He2017beta}. The polar codes constructed by these existing methods \cite{arikan2009channel, mori2009performance, trifonov12efficient, He2017beta} show good performance with SCD in the long code length regime. However, the error-rate performance degrades for the polar codes optimized for the SCD performance when the code length is moderate to short.
In \cite{tal2015list}, the SC-list (SCL) decoder is proposed which generates a list of decoding candidates during the SCD process to overcome the performance degradation in SCD. It is also shown in \cite{tal2015list} that when the list size $L$ is sufficiently large, SCL decoding (SCLD) performance approaches the maximum likelihood (ML) decoding performance.
Thus, the performance of SCLD is closely related to the minimum distance of polar codes. However, when polar codes optimized for the SCD performance are decoded with an SCL decoder, they have limited performance due to their poor minimum distance characteristics. It is shown in \cite{arikan2009channel} that polar codes designed by selecting information bits that maximize the minimum distance turn out to be Reed-Muller (RM) codes. By utilizing the information bit decision rule in RM codes, a hybrid construction method called RM-polar code is proposed in \cite{li2014rm}.

There have been notable efforts \cite{wang2016parity, trifonov2016polar, trifonov2017randomized, arikan2019sequential} to improve the SCLD performance of polar codes by concatenation with other types of channel codes. In \cite{wang2016parity}, a heuristic approach for designing parity-check concatenated polar codes is proposed by concatenating a polar code with an outer parity-check code. The polar subcodes of extended Bose-Chaudhuri-Hocquenghem (eBCH) codes and the concept of dynamic frozen bits are proposed in \cite{trifonov2016polar}. Due to the large minimum distance characteristics of eBCH codes, the polar subcodes in \cite{trifonov2016polar} improve the SCLD performance with a large $L$. In \cite{trifonov2017randomized}, the randomized construction for polar subcodes is proposed by reducing the number of low-weight codewords to improve the SCLD performance. Meanwhile,  the polarization-adjusted convolutional (PAC) codes are proposed by Arikan in \cite{arikan2019sequential}. In a PAC code, a rate-1 convolutional code is concatenated with a polar code.  It is shown in \cite{arikan2019sequential} that  PAC codes achieve the dispersion bound when decoded with the Fano decoder \cite{fano1963heuristic}. However, the Fano decoder has high and unpredictable decoding latency as compared to the list decoding of these codes \cite{rowshan2021polarization}. The existing construction schemes \cite{wang2016parity, trifonov2016polar, trifonov2017randomized, arikan2019sequential} for polar codes focus on designing the structure of message bits (i.e., information and frozen bits). In this work, we propose a novel construction scheme for polar codes by using the puncturing and extending techniques \cite{wicker1994error} aiming to improve the SCLD performance.

 It has been shown that by carefully designing the puncturing/extending patterns, the designed codes become more flexible/stronger while retaining or improving the property of the base code \cite{grassl2004new, hsu2008capacity, han2011finding}. In \cite{li2019optimal, zhao2021novel, Han2022rate}, the design criteria for puncturing patterns in polar codes are proposed to minimize the degradation of the reliability in the information bits.  In particular, a puncturing scheme called worst quality puncturing (WQP) is proposed in \cite{li2019optimal}, where puncturing patterns are designed to minimize the bit-channel reliability loss in the message bits due to the zero reliability channels associated with the punctured bits. In \cite{zhao2021novel}, the information set approximation puncturing (ISAP) algorithm is proposed to reduce the influence of punctured bits on the information set $\mathcal I$ by introducing guard bits. Meanwhile, a metric called reliability score is defined in \cite{Han2022rate} to measure the reliability of information bits. Then, puncturing patterns are designed to maximize the minimum reliability score of the information bits. Meanwhile, there have been studies \cite{chen2013hybrid, saber2015incremental, zhao2018adaptive} on extended polar codes in which the extending patterns are designed to improve the SCD performance. In \cite{chen2013hybrid}, the extended polar codes are designed by repeating the information bits with weak reliability in SCD. In \cite{saber2015incremental}, extended bits are selected among the nodes in the polar code graph (PCG). The decoding performance of extended polar codes is improved in \cite{zhao2018adaptive}, by expanding the size of the polarizing matrix in which extended bits are inserted.

In the proposed design scheme, a polar code having a large minimum distance \cite{arikan2009channel}, i.e., RM code, is first set as a base code. Then, the coded bits in the base code are punctured to obtain a degree of freedom in transmission called \emph{transmission holes}. The puncturing patterns are designed in such a way that the loss in the minimum distance due to puncturing is minimized. To this end, we derive some minimum distance properties of punctured RM codes, which allows us to develop an efficient algorithm for designing puncturing patterns. The transmission holes are filled with additional parity bits generated by extending the punctured RM code. In this way, the designed code has the same length and rate as those of the base code. In the proposed code design scheme, the extending patterns are designed to maximize the SCLD performance. However, to the best of our knowledge, no efficient mathematical tool exists to analyze and maximize the SCLD performance for a practical list size (i.e., $4 \le L \le 16$). To achieve the design goal, the reinforcement learning (RL) technique is utilized in this work, which has achieved many successes in the area of channel coding, e.g., code design \cite{huang2019ai, liao2021construction} and decoding algorithm \cite{carpi2019reinforcement, doan2020decoding}. In particular, an RL-based constructor-evaluator framework is proposed in \cite{huang2019ai} for designing several types of codes, e.g., linear block codes and polar codes. In this framework, the constructor learns to construct codes with improved error-rate performance which is then evaluated by the evaluator. In \cite{liao2021construction}, the problem of selecting the information set $\mathcal I$ of polar codes is modeled with a maze-traversing game and solved using RL techniques. The designed codes with RL techniques \cite{huang2019ai, liao2021construction} show improved performance compared to the codes designed with the existing constructions without resorting to RL techniques.

To apply RL for designing extending patterns, we first formulate the design of extending patterns for polar codes as a finite Markov decision process (FMDP). The proposed code design scheme selects some nodes in PCG, called \emph{extended nodes}, and then determines how many times the bit values at the extended nodes are repeated. Thus, the extending pattern is defined as the numbers of repetition of all nodes in PCG. Since the bit values at all nodes in PCG are produced during the encoding, it does not require extra complexity to make the extended bits. In the meantime, SCLD can be performed by simply initializing the reliabilities at the extended nodes with the channel outputs, which greatly simplifies SCLD of the extended polar codes.

In the RL-based design of extending patterns, the extending pattern and the indices of new nodes to be extended are defined as the state and the action in RL, respectively. In particular, this work implements the design by utilizing the deep-$Q$ learning (DQL) technique \cite{mnih2015human}. In doing so, it will be observed that the sizes of state and action spaces in DQL grow prohibitively fast, which requires intolerably expensive complexity and long design time. To turn around the technical challenge, we smartly introduce a multi-stage learning and state reduction scheme which not only greatly expedite the design process but also significantly improve the performance of the designed polar codes. Finally, it will be demonstrated that the combination of proposed puncturing and extending schemes leads to a novel polar code design scheme, which allows notable improvement of the SCLD performance as compared to the existing code construction schemes.

The main contributions in this paper are summarized as follows:
\begin{itemize}
 \item We propose a novel polar code design scheme using puncturing and extending. It is demonstrated that the designed polar codes have considerable improvement in SCLD performance over existing polar codes including the 5G NR and modified polar codes (e.g., PAC codes).

 \item We derive some minimum distance properties of punctured RM codes and propose a puncturing scheme for RM codes. The proposed puncturing scheme minimizes the decrease of the minimum distance caused by puncturing. In addition, we discuss how to apply the puncturing scheme to other types of polar codes.

  \item An RL-based extending scheme is proposed aiming to maximize the SCLD performance, and we introduce a novel idea to significantly reduce the sizes of state and action spaces based on the analysis of log-likelihood ratio (LLR) message propagation in SCD. 

\end{itemize}

The remainder of the paper is organized as follows. In Section \ref{Sec:Pre}, the basic concepts of the polar codes are summarized. In Section \ref{Sec:Prop}, we introduce the proposed design scheme based on the puncturing and extending of polar codes. In addition, the proposed puncturing scheme for RM and polar codes is described in detail. In Section \ref{Sec:Ext}, we formulate the design scheme for extending patterns of polar codes using RL. To reduce the training complexity of RL, we introduce a multi-stage learning and a novel idea to reduce the state and action spaces based on the analysis of SCD. In Section \ref{Sec:Sim}, the error-rate performance of proposed polar codes is evaluated and compared with polar codes designed with the existing design schemes. Finally, we conclude this work in Section \ref{Sec:Con}.

\renewcommand{\thefootnote}{\arabic{footnote}}

\section{Polar Codes Basics}\label{Sec:Pre}
\subsection{Polar Codes} \label{Sec:Polar}
An $(N, K)$ polar code with codeword length $N = 2^m$ and message length of $K$ is constructed with a generator matrix $G = F^{\otimes m}$, where $F$ is a polarization kernel\footnote[1]{As shown in \cite{korada2010polar}, various types of polarization kernels can be used to design polar codes. In general, the binary kernel matrix, i.e. $F = \left[\begin{smallmatrix} 1 & 0 \\ 1 & 1 \end{smallmatrix}\right]$, is widely used.} and $\otimes^m$ denotes the $m$-fold Kronecker product of the matrix. A polar code is defined with two sets of bit indices: 1) information set $\mathcal I \subset \{0, 1, \ldots, N-1\}$ consisting of the indices of information bits, and 2) the frozen set $\mathcal F = \mathcal I^c$ consisting of the indices of frozen or parity bits. A codeword $\mathbf c$ of a polar code can be generated as $\mathbf c = \mathbf u \cdot G$, where a binary message vector $\mathbf u = \{ u_0, u_1, \ldots, u_{N-1}\}$ is partitioned into the information and frozen bits denoted by $u_{\mathcal I}$ and $u_{\mathcal F}$, respectively. The structure of the polar code can be represented with a bipartite graph, and we name it \textit{polar code graph} (PCG). In Fig. \ref{Fig:PG}, the PCG of $(8, 5)$ polar code is shown with the information set $\mathcal I = \{3, 4, 5, 6, 7\}$ and the frozen set $\mathcal F = \{0,1,2\}$, which are represented by the filled and shaded circles at stage 0, respectively. Each node in PCG is designated by a pair of indices $(i, j)$, where $i$ and $j$ indicate the message and stage indices, respectively. The set of all nodes in PCG is denoted by $\mathcal N = \{(i, j): 0 \le i \le N - 1, 0 \le j \le m\}$.

There exist nodes in PCG with infinite LLR values due to the frozen bits, and the nodes are referred to as \emph{infinite nodes}. The set of infinite nodes, called \emph{infinity set}, is denoted by $\mathcal N_\infty$. The infinity set at stage 0, denoted by $\mathcal N_{\infty, 0}$, is given by
\[
  \mathcal N_{\infty, 0} = \{(i, 0) | i \in \mathcal F\}.
\]
Then, the infinity set at stage $j > 0$, denoted by $\mathcal N_{\infty, j}$, can be recursively obtained from $\mathcal N_{\infty, j - 1}$ as
\begin{align}
  \mathcal N_{\infty, j}  = \Bigl \{(i, j) | & b_{i, j - 1} = 0, \nonumber \\
  &  \{(i + 2^{j - 1}, j - 1), (i, j - 1) \} \subseteq \mathcal N_{\infty, j - 1} \Bigr \} \nonumber \\
   \cup  \Bigl \{(i, j) | & b_{i, j - 1} = 1, (i, j - 1) \in \mathcal N_{\infty, j - 1} \Bigr \}, \label{Eq:InfNode}
\end{align}
where $\mathbf b_i = \{b_{i, 0}, b_{i, 1}, \ldots, b_{i, m - 1}\}$ is the binary representation vector of the index $i$, i.e., $i = \sum_{j = 0}^{m - 1} b_{i, j} 2^{j}$ for $b_{i, j} \in \{0, 1\}$. Using the definition of $\mathcal N_{\infty, j}$ in \eqref{Eq:InfNode}, the infinity set of a polar code is given by
\begin{equation} \label{Eq:InfSet}
  \mathcal N_\infty = \bigcup_{j = 0}^m \mathcal N_{\infty, j}.
\end{equation}
In Fig. \ref{Fig:PG}, the shaded circles indicate the infinite nodes, and $\mathcal N_\infty = \{(0, 0), (1, 0), (2, 0), (0, 1), (1, 1)\} $.

When some coded bits are punctured or erased (this work considers only the punctured bits), the channel outputs of punctured bits have zero reliability which is propagated to the intermediate nodes in PCG through $\boxplus$ and $+$ operations during SCD/SCLD. The nodes with zero reliability in PCG are called \emph{zero nodes}, and {these nodes degrade} the decoding performance. The set of nodes with zero reliability is referred to as \emph{zero set} and denoted by $\mathcal N_z$. The set of zero reliability nodes at stage $m$, denoted by $\mathcal N_{z, m} \subseteq \mathcal N_z$, is given by 
\[
 \mathcal N_{z, m} = \{(i, m) | i \in \mathcal P\}.
\]
At stage $j < m$, the nodes of zero reliability can be recursively determined with the initial set $\mathcal N_{z, m}$. That is, node $(i, j)$ for $b_{i, j} = 0$ will have zero reliability due to the $\boxplus$ operation when either node $(i, j)$ or $(i + 2^j, j)$ has zero reliability. For example, node (0, 2) in Fig. \ref{Fig:PG} has zero reliability when either (0, 3) or (4, 3) has zero reliability. Meanwhile, for $b_{i, j} = 1$, node $(i, j)$ will have zero reliability when both nodes $(i, j)$ and $(i - 2^j, j)$ have zero reliability. For example, when both nodes (0, 3) or (4, 3) have zero reliability, node (4, 2) ends up  with zero reliability. The recursion from stage $j + 1$ to stage $j$ can be summarized as follows:
\begin{equation}
    \mathcal N_{z, j} = \mathcal N'_{z, j} \setminus \mathcal N_\infty, \label{Eq:ZSet}
\end{equation}
where
\begin{align*}
  \mathcal N'_{z, j} = &\left \{(i, j) | b_{i, j} = 1, \{(i - 2^j, j + 1), (i, j + 1)\} \subseteq \mathcal N_{z, j + 1} \right \} \\
   &\cup \left \{(i, j) | b_{i, j} = 0, (i - 2^j, j + 1) \in \mathcal N_{z, j + 1} \right \} \\
   &\cup \left \{(i, j) | b_{i, j} = 0, (i, j + 1) \in \mathcal N_{z, j + 1} \right \}.
\end{align*}
Note that in \eqref{Eq:ZSet}, the nodes in the infinity set must be excluded from the zero set. Using the zero set at stage $j$ in \eqref{Eq:ZSet}, the set of zero reliability nodes is given by
\[
  \mathcal N_z = \bigcup_{j = 0}^m \mathcal N_{z, j}.
\]
For each information bit in PCG, a tree, referred to as an \emph{information tree}, can be established in such a way that the tree has the information bit as its root and the nodes in the path from the root with increasing stage indices as the descendants. In Fig. \ref{Fig:PG}, the information tree for the information bit $u_4$ is depicted by the shaded area.

\begin{figure}[t]
 \centerline{\includegraphics[width=0.9\columnwidth]{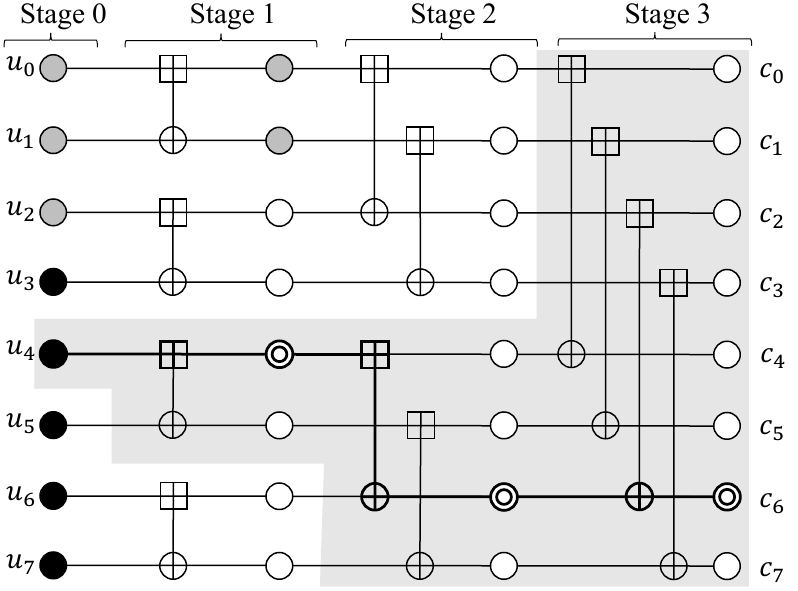}}
	\caption{Polar code graph (PCG) of an (8,5) polar code where nodes are depicted by circles and double circles. The filled and shaded circles at stage 0 indicate the nodes for the information and frozen bits, respectively.} \label{Fig:PG}
\end{figure}

After polar encoding, the codeword $\mathbf c = (c_0, c_1, \ldots, c_{N - 1})$ is modulated into $\mathbf x = (x_0, x_1, \ldots, x_{N - 1})$ and transmitted through a binary-input discrete memoryless channel. The receiver receives the channel output $\mathbf y = (y_0, y_1, \ldots, y_{N-1})$ and performs the decoding. In SCD of polar codes \cite{arikan2009channel}, the information bits are estimated based on the sequential message passing over PCG. At decoding step $i$, the LLR of information bit $i$, denoted by $\Lambda_i$, is obtained using the received channel output $\mathbf y$ and the estimates of the past decisions $\hat {\mathbf u}^{i-1}_0$ as
\begin{equation} \label{Eq:L_i}
    \Lambda_i = \log \left( \frac{\Pr(u_i = 1 | \hat {\mathbf u}^{i-1}_0, \mathbf y)}{\Pr(u_i = 0 | \hat {\mathbf u}^{i-1}_0, \mathbf y)} \right)
\end{equation}
where $\hat{\mathbf u}^{i-1}_0 = \{ \hat{u}_0, \hat{u}_1, \ldots, \hat{u}_{i-1} \}$, and $\hat{u}_j$ is the estimate of the $j$-th message bit for $0 \le j \le i - 1$. The bit decision on the LLR of the information bit is taken as
\begin{equation} \label{Eq:u_i}
\hat u_i =
\begin{cases}
	1   &\text{if  }  \Lambda_i >0, \\
	0   &\text{otherwise}.
\end{cases}
\end{equation}
Meanwhile, for a frozen bit (i.e., $i \in \mathcal F$), a predetermined value is set as the message value and is known to both the transmitter and receiver. The SCLD in \cite{tal2015list} improves the error-correction capability of SCD by tracking $L$ parallel paths during decoding. Note that the performance of SCLD approaches the ML performance when the list size $L$ is sufficiently large \cite{tal2015list}.

\subsection{Reed Mullar Codes}
An RM$(r, m)$ code is a linear block code defined by length $N=2^m$, dimension $K = \sum_{k = 0}^r {m \choose k}$, and degree $0 \le r \le m$. An RM code can be interpreted as a polar code with the information set $\mathcal I$ in such a way that for $i \in \mathcal I$ and $j \in \mathcal F$,
\[
  w_r(i) > w_r(j),
\]
where $w_r(i)$ is the Hamming weight of the $i$-th row in $G$. Hereafter, the Hamming weight of rows in $G$ will be simply called row weight for short. It can be noticed that the indices in $\mathcal I$ are determined in the decreasing order of row weights \cite{arikan2009channel, li2014rm}. Note that the minimum distance of the polar codes is obtained by finding the minimum row weight of $G$ among the rows corresponding to the information bit indices as follows \cite{korada2009polar}:
\[
  d_{\min} = \min \{ 2^{w_b (i)} | i \in \mathcal I \} = \min \{ w_r (i) | i \in \mathcal I \},
\]
where $w_b(i)$ is the Hamming weight of the binary representation of index $i$, i.e., $w_b(i) = \sum_{j = 0}^{m-1} b_{i, j}$. Since the information bit indices of RM codes correspond to the rows with the largest weights, RM codes have the largest minimum distance among polar codes (i.e., $d_{\min}=2^{m-r}$).

\subsection{Puncturing and Extending of Polar Codes}

For a base code with rate $R = K / N$, puncturing removes a set of $N_P$ coded bits, denoted as $\mathcal P = \{p_0, p_1, \ldots, p_{N_P - 1}\}$ to increase the code rate to $R_P = K / (N - N_P)$ where $0 \le p_i \le N - 1$ indicates the index of the $i$-th punctured bit. The punctured bits are not transmitted, and no information is obtained from the channel about these bits. Thus, the LLR values of the punctured bits are set as zero at the decoder. For example, if two bits are punctured with $\mathcal P = \{0, 1\}$ in the PCG shown in Fig. \ref{Fig:PG}, only the coded bits $x_2, x_3, x_4, x_5, x_6,$ and $x_7$ are transmitted through the channel. The channel output for punctured bits is set as zero, i.e., $y_0 = y_1 = 0$. In general, it is shown in\cite{shin2013design} that the channel splitting for the zero-capacity punctured channels results in zero-capacity synthetic channels for some message bits called \emph{incapable} bits. Furthermore, it is noted that the number of incapable bits is equal to the number of punctured bits.

On the contrary, the extending technique produces $N_E$ additional parity bits as a part of an extended codeword, whose code rate gets reduced to $R_E = K/(N + N_E)$. The extending in this work simply repeats bit values at \emph{any} nodes in PCG as the additional parity bits, which is also considered in \cite{saber2015incremental}. The repetitions of nodes are described with the extension matrix defined in Definition \ref{Def:EM}. Note that the bit values of all the nodes in a PCG can be obtained without any additional computations during the encoding process. It should also be noted that $e_{i, 0} > 0$ and $e_{j, m} > 0$ indicate the repetitions of message bit, $u_i$ and coded bit $c_{j}$, respectively. The selected nodes for extending will be called \emph{extended} nodes, and the set of extended nodes is called \emph{extending set} and denoted by $\mathcal E$. 
\begin{definition}[Extension Matrix] \label{Def:EM}
The extension matrix is an $N \times (m + 1)$ matrix given by
\[
  E =
  \begin{bmatrix}
    e_{0, 0} & e_{1, 1} & \cdots & e_{0, m} \\
    e_{1, 0} & e_{1, 1} & \cdots & e_{1, m} \\
     \vdots & \vdots  & \ddots & \vdots  \\
    e_{N - 1,0} & e_{N - 1, 1} & \cdots & e_{N - 1, m}
  \end{bmatrix}
\]
where $e_{i, j} \in \mathbb N_0$ is the number of repetitions for the bit value at node $(i, j)$.
\end{definition}

\begin{figure}[t]
 \centerline{\includegraphics[width=0.99\columnwidth]{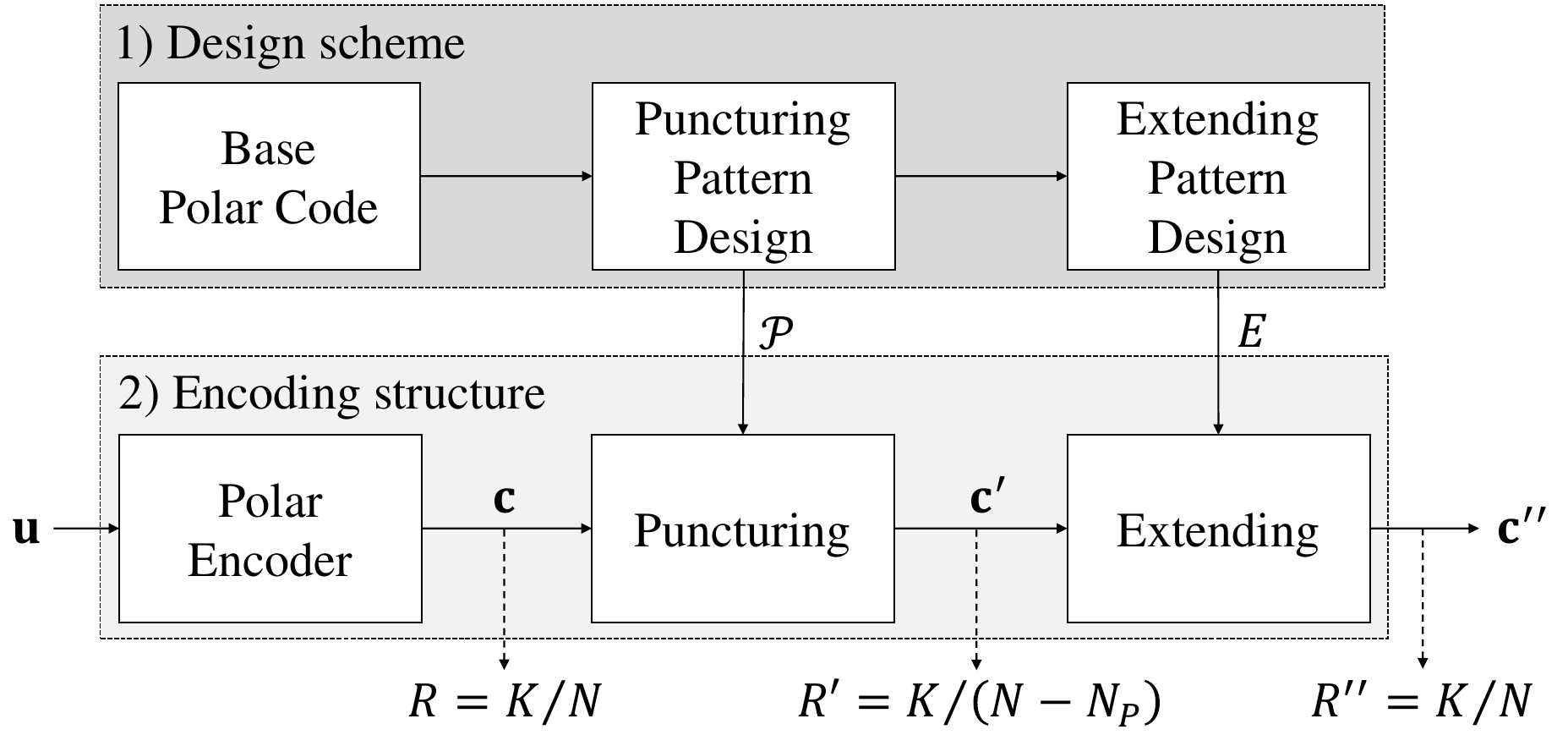}}
	\caption{Block diagram of the proposed design scheme for polar codes using puncturing and extending.} \label{Fig:PDP}
\end{figure}

\section{Polar Code Design based on Puncturing and Extending}\label{Sec:Prop}

\subsection{Design Strategy}\label{SubSec:Design}
The proposed design scheme is illustrated in Fig. \ref{Fig:PDP} where we consider an RM code of length $N$ as the base code. For the given base code, $N_P$ punctured bit indices in the puncturing pattern $\mathcal P$ are determined in such a way to minimize the decrease of minimum distance after puncturing. Then, for the punctured polar codes, we decide $N_E$ extended bits by designing the extension matrix, $E = [e_{i, j}]$ for $N_E = \sum_{i, j} e_{i, j}$. The proposed extending scheme finds the extension matrix, $E$ to maximize the SCLD performance. It is assumed that the numbers of punctured and extended bits are the same, i.e., $N_E = N_P = N_h$, which makes the code rate of the designed code $R''$ in Fig. \ref{Fig:PDP} the same as that of the base code $R$. The number $N_h$ will be referred to as the number of transmission holes.

For the designed puncturing and extending patterns, i.e., $\mathcal P$ and $E$, the encoding is performed in such a way that a message bit vector $\mathbf u$ is first encoded into a codeword of the base code, denoted by $\mathbf c$ in Fig. \ref{Fig:PDP}. Then, $N_h$ parity bits are punctured from $\mathbf c$, which creates $N_h$ transmission holes in $\mathbf c'$, resulting in a degree of freedom in the transmission. After that, the codeword of the proposed design, $\mathbf c''$ is made by adding the extended bits for maximizing the SCLD performance. The details of the proposed puncturing and extending schemes are shortly discussed in Sections \ref{SubSec:Punc} and \ref{Sec:Ext}, respectively.

During the design process, the number of transmission holes $N_h$ plays an important role in the SCLD performance of the designed polar code. If a large number of bits are punctured (i.e., when $N_h$ is large), the minimum distance of the code deteriorates too much from that of the base code. Whereas, for small $N_h$, the degree of freedom in the design is limited and it is hard to obtain sufficient SCLD performance gain by extending. It is noticed that when we reach the limit of the performance improvement through extending, the punctured bits will be recovered by extending since the performance gain by additional extending is less than the gain achieved by recovering the punctured bits. Thus, for designing good puncturing and extending patterns, we gradually increase $N_h$ until the extending recovers the punctured bits.

\subsection{Puncturing Pattern Design} \label{SubSec:Punc}
The minimum distance of a code decreases by puncturing since it removes some of the coded bits and the associated structural connections. In this section, we propose a puncturing scheme tailored for RM codes which provides puncturing patterns minimizing the loss of the minimum distance after puncturing. We denote a puncturing pattern of length $\ell$ by $\mathcal P^{(\ell)} = \{p_0, p_1, \ldots, p_{\ell - 1} \}$, where $\mathcal P^{(0)} = \emptyset$, and $p_i$ is the index of the $i$-th punctured bit. For simplicity, the superscript $\ell$ will be omitted when it is not necessary. Let us denote the minimum distance of a code and the number of codewords with the minimum non-zero weight by $d_{\min}$ and $A_{d_{\min}}$, respectively. Then, the minimum distance and the number of codewords with the minimum weight after puncturing are denoted by $d_{\min} (\mathcal P^{(\ell)})$ and $A_{d_{\min}} (\mathcal P^{(\ell)})$, respectively, when the puncturing is performed with $\mathcal P^{(\ell)}$.

In the proposed scheme, the design problem is formulated as the selection of $\mathcal P^{(\ell)}$ which maximizes the minimum distance and is defined as
\begin{equation} \label{Eq:minNode1}
\mathcal P' = \arg\max_{p \in \mathcal B^{(\ell)}}  d_{\min} (\mathcal P^{(\ell-1)}\cup \{p\}),
\end{equation}
where $\mathcal B^{(\ell)} = \{0, 1, \ldots, N - 1\} \backslash \mathcal P^{(\ell - 1)}$ is the index set of the unpunctured coded bits after puncturing with $\mathcal P^{(\ell - 1)}$. If the cardinality of $\mathcal P'$ is unity, i.e., $|\mathcal P'| = 1$, $\mathcal P'$ contains only the $\ell$-th punctured bit, which simply leads to $\mathcal P^{(\ell)} =\mathcal P^{(\ell - 1)} \cup \mathcal P' $. Otherwise, among the multiple candidates in $\mathcal P'$, $p_{\ell}$ is selected to minimize the number of codewords with the minimum weight as
\begin{equation} \label{Eq:minNode2}
  p_{\ell}=\arg\min_{p \in \mathcal P'}(A_{d_{\min}}  (\mathcal P^{(\ell - 1)} \cup \{p\})).
\end{equation}
The design steps in \eqref{Eq:minNode1} and \eqref{Eq:minNode2} are repeatedly performed until the number of punctured bits $\ell$ reaches $N_p$. To conduct the design steps in an efficient way, we first derive minimum distance properties for RM codes which are carefully harnessed to devise an algorithm for finding puncturing patterns.

An RM$(r, m)$ code can be represented by a set of polynomials \cite{wicker1994error}. Let us denote $P_{r, m}$ as the set of polynomials of $m$ variables $\mathbf x = (x_0, x_1, \ldots, x_{m - 1})$ having degree less than or equal to $r$. In addition, let $f(\mathbf b)$ be the evaluation of the polynomial $f(\mathbf x) \in P_{r,m}$ at $\mathbf b \in \mathbb F_2^m$. Then, for each polynomial $f(\mathbf x) \in P_{r,m}$, there exists a unique codeword $\mathbf c = (c_0, c_1, \ldots, c_{n - 1})$ in the RM$(r, m)$ code such that each coded bit $c_i$ of $\mathbf c$ is given by the evaluation of $f(\mathbf x)$ at $\mathbf x = \mathbf b_i$, where $\mathbf b_i = (b_{i, 0}, b_{i, 1}, \ldots, b_{i, m - 1})$ is the binary representation vector of the index $i$, i.e., $i = \sum_{j = 0}^{m - 1} b_{i, j} \; 2^{j}$ for $b_{i, j} \in \{0, 1\}$. Let $\mathcal C_{\min}$ be the set of codewords with the minimum weight, i.e., $\mathcal C_{\min} = \{\mathbf c |  w_H(\mathbf c) = d_{\min}, \mathbf c \in {\rm RM}(r, m) \}$. Then, a coded bit $c_i$ for $\mathbf c \in \mathcal C_{\min}$, is obtained by evaluating $f(\mathbf x) \in P_{r,m}$ satisfying the following equation \cite{macwilliams1977theory}:
\begin{equation} \label{Eq:minf}
  f(\mathbf x) = \prod_{k=0}^{r - 1} \left(\xi_{k} + \sum^{m - 1}_{v=0} z_{k, v} x_v \right),
\end{equation}
where
\[
  Z = \left[\begin{smallmatrix} z_{0, 0}  & \cdots & z_{0, m - 1}  \\ \vdots & \ddots & \vdots \\ z_{r - 1, 0} & \cdots & z_{r - 1, m - 1} \end{smallmatrix}\right] \in \Psi_{r,m},
\]
and $\Psi_{r, m}$ is a set of $r \times m$ binary matrices with rank $r$ and $\boldsymbol \xi = \{\xi_0, \ldots, \xi_{r - 1}\} \in \mathbb F^r_2$. That is, a coded bit $c_i$ of $\mathbf c \in \mathcal C_{\min}$, is given by
\[
  c_{i} = f(\mathbf b_i) = \prod_{k = 0}^{r - 1} \left(\xi_{k} + \sum^{m - 1}_{v = 0} z_{k, v} b_{i, v} \right).
\]

The minimum distance of punctured code depends on the weight distribution of its base code and the puncturing pattern. It will be shown in Theorem \ref{Thm:mw} that the minimum weight codewords of punctured code are obtained by deleting coded bits in $\mathbf c \in \mathcal C_{\min}$ if the number of punctured bits $\ell$ is less than or equal to the minimum distance of the base RM code. That is, some codewords only in $\mathcal C_{\min}$ will be the codewords of punctured code with the minimum weight $d_{\min}(\mathcal P^{(\ell)})$, regardless of the puncturing pattern $\mathcal P^{(\ell)}$ if the number of punctured bits, $\ell$ is less than or equal to $d_{\min}$.

\begin{lemma} \label{Lm:O}
  For any coded bit index $i \in \{0, 1, \ldots N - 1\}$, there exists at least one minimum weight codeword $\mathbf c \in \mathcal C_{\min}$ such that $c_i$ is one, i.e., $c_i = 1$.
\end{lemma}
\begin{proof}
Note that each polynomial in $P_{r, m}$ uniquely corresponds to a codeword in an RM$(r, m)$ code. Then, for a polynomial
\begin{equation} \label{Eq:fx}
  f(\mathbf x) = \prod_{k=0}^{r - 1} (1 + b_{i, k} + x_{k}),
\end{equation}
there exists the corresponding codeword, denoted by $\mathbf c$ for which the $i$-th coded bit $c_i$ is given by $f(\mathbf b_i) = 1$.

Meanwhile, the polynomial $f(\mathbf x)$ can be understood as the one in \eqref{Eq:minf} with $\xi_k = 1 + b_{i, k}$ and 
\[
  Z =
  \begin{bmatrix}
     I_r  & \mathbf 0_{r \times (m - r)}
  \end{bmatrix}
\]
where $I_r$ and $\mathbf 0_{r \times (m - r)}$ are the $r \times r$ identity matrix and the $r \times (m - r)$ all zero matrix, respectively. Note that $Z$ has its rank of $r$, which tells that the codeword $\mathbf c$ is in $\mathcal C_{\min}$. Thus, the lemma follows.
\end{proof}

\begin{theorem} \label{Thm:mw}
  The minimum weight codewords of a punctured RM$(r, m)$ code are obtained by deleting coded bits of codewords only in $\mathcal C_{\min}$ if the number of punctured bits $\ell \le d_{\min} = 2^{m-r}$.
\end{theorem}
\begin{proof}
 For a codeword $\mathbf c \in {\rm RM}(r, m)$, let us denote the weight and the decrease of weight after puncturing with $\mathcal P = \{p_0, p_1, \ldots, p_{\ell - 1}\}$ by $w_H(\mathbf c, \mathcal P)$ and $\Delta w_H (\mathbf c, \mathcal P) = w_H(\mathbf c) - w_H(\mathbf c, \mathcal P)$, respectively. According to Lemma \ref{Lm:O}, there exists at least one codeword, $\mathbf c \in \mathcal C_{\min}$ such that $w_H(\mathbf c, \mathcal P) < d_{\min}$ after puncturing. Meanwhile, for a codeword $\mathbf c'$ of weight $w_H(\mathbf c') \ge 2 d_{\min} = 2^{m-r+1}$, $w_H (\mathbf c', \mathcal P) \ge d_{\min}$ since the number of punctured bits $\ell$ is less than or equal to $d_{\min}$. Thus, a codeword of weight $w_H(\mathbf c') \ge 2 d_{\min}$ cannot be a minimum weight codeword of the punctured RM($r, m$) code.

In \cite{berlekamp1969restrictions}, it is shown that for a codeword $\mathbf c'' \in {\rm RM} (r, m)$ of weight $w_H(\mathbf c'') < 2 d_{\min} = 2^{m-r+1}$, its weight can be expressed as
\[
	w_H(\mathbf c'') = 2^{m-r+1} - 2^{m-r+1-\eta}
\]
for $\eta \in \{1, 2, \ldots, m - r + 1\}$. We denote the set of codewords having the same weight by $\mathcal C_{\eta}= \{\mathbf c| w_H(\mathbf c) =  2^{m-r+1} - 2^{m - r + 1 -\eta}\}$. Note that $\mathcal C_{1}= \mathcal C_{\min}$. It is also shown in \cite{kasami1970weight} that a codeword $\mathbf c \in \mathcal C_{\eta}$ can be expressed as a sum of $\nu \le \eta$ minimum weight codewords, i.e., $\mathbf c = \sum_{j = 1}^{\nu} \mathbf c_j $ where $\mathbf c_j \in \mathcal C_{\min}$. Then, the weight reductions of $\mathbf c_j$ and $\mathbf c$ satisfy the following inequality:
\begin{equation} \label{Eq:wh}
   \Delta w_H (\mathbf c, \mathcal P) \le \nu \max_{j \in \{1, 2, \ldots, \nu\}} \Delta w_H (\mathbf c_j, \mathcal P).
\end{equation}

If there exists a codeword $\mathbf c \in \mathcal C_{\eta}$, for $\eta \ge 2$ which becomes a codeword of minimum weight after puncturing, i.e, $w_H(\mathbf c, \mathcal P) = d_{\min}(\mathcal P)$, then the following inequality must be fulfilled: 
\begin{align}
   d_{\min}(\mathcal P)
      &  \le \min_{j \in \{1, \ldots, \nu\}} w_H(\mathbf c_j, \mathcal P) \nonumber \\
      &  = d_{\min} - \max_{j \in \{1, \ldots, \nu\}} \Delta w_H (\mathbf c_j, \mathcal P) \nonumber \\
      &  \le d_{\min} - \frac{\Delta w_H (\mathbf c, \mathcal P)}{\nu} \label{Eq:dminP1} \\
      &  \le d_{\min} - \frac{\Delta w_H (\mathbf c, \mathcal P)}{\eta}, \label{Eq:dminP}
\end{align}
where the inequality in \eqref{Eq:dminP1} is due to \eqref{Eq:wh}. The equality, $d_{\min}(\mathcal P) = w_H(\mathbf c) - \Delta w_H(\mathbf c, \mathcal P)$ and the inequality in \eqref{Eq:dminP} give us the following inequality:
\begin{align}
     \Delta w_H (\mathbf c, \mathcal P)&\ge \frac{w_H(\mathbf c) - d_{\min}}{(1 - \eta^{-1})} \nonumber \\
      & = 2^{m-r} \frac{1- 2^{-\eta}}{1- \eta^{-1}}
       > 2^{m-r} = d_{\min} \label{Eq:Dr}
\end{align}
where the inequality in \eqref{Eq:Dr} is due to $2^{\eta} > \eta$, $\forall \, \eta \in \mathbb N$. Thus, if a minimum weight codeword of punctured RM code comes from  $\mathcal C_\eta$ for $\eta \ge 2$, the inequality in \eqref{Eq:Dr} must be satisfied. However, the number of punctured bits $\ell \le d_{\min}$, and thus, the inequality in \eqref{Eq:Dr} cannot be fulfilled. In summary, it can be concluded that the minimum weight codewords after puncturing are only from $\mathcal C_{\min}$.
\end{proof}

Due to Theorem \ref{Thm:mw}, we can find the minimum distance of a punctured RM code by tracking the changes in the weights of the codewords in $\mathcal C_{\min}$ when the number of punctured bits $\ell$ is less than or equal to $d_{\min}$.

The amount of reduction in the minimum distance of the base code after puncturing with $\mathcal P^{(\ell)} = \{p_0, p_1, \ldots, p_{\ell - 1}\}$ is denoted as $\Delta d_{\min} (\mathcal P^{(\ell)}) = d_{\min} - d_{\min}(\mathcal P^{(\ell)})$. Since the minimum distance of a punctured code, i.e., $d_{\min}(\mathcal P^{(\ell)})$ is determined by the weights of codewords in $\mathcal C_{\min}$ for $\ell \le d_{\min}$, the set of minimum weight codewords is defined as
\begin{equation}
  \mathcal C_{\min} (\mathcal P^{(\ell)}) = \left\{\mathbf c   \in \mathcal C_{\min} \middle | \Delta w_H (\mathbf c, \mathcal P^{(\ell)}) = \Delta d_{\min} (\mathcal P^{(\ell)}) \right \}.
\end{equation}
When a puncturing pattern $\mathcal P^{(\ell - 1)}$ is extended to $\mathcal P^{(\ell)} = \mathcal P^{(\ell - 1)} \cup \{p'\}$ with a new punctured bit index $p' \not \in \mathcal P^{(\ell - 1)}$, the minimum distance of the punctured code with $\mathcal P^{(\ell)}$  is determined as follows:
\begin{multline}
d_{\min}(\mathcal P^{(\ell)})  \\
  =\begin{cases}
    d_{\min}(\mathcal P^{(\ell - 1)}) - 1, & \text{if }\,  \exists\, \mathbf c \in \mathcal C_{\min}(\mathcal P^{(\ell - 1)}) \text{ s.t. } c_{p'} = 1,\\
    d_{\min}(\mathcal P^{(\ell - 1)}),      & \text{otherwise}.
  \end{cases}
\end{multline}
Note that to avoid the decrease of minimum distance due to the additional punctured bit $c_{p'}$, it is necessary to decide the bit index $p'$ in such a way that for all $\mathbf c \in \mathcal C_{\min}(\mathcal P^{(\ell - 1)})$, the bit value at $p'$ is zero, i.e., $c_{p'} = 0$. Theorem \ref{Thm:Punc} tells whether such an index $p'$ exists or not.

\begin{theorem} \label{Thm:Punc}
For a set of indices $\mathcal I = \{ i_0, i_1, \ldots, i_{\ell - 1} \}$ and $\ell \le d_{\min}$, there exists a codeword $\mathbf c \in \mathcal C_{\min}$ for RM($r, m$) code such that the coded bits at all the indices in $\mathcal I$ are ones, i.e., $c_i = 1$, $\forall i \in \mathcal I$, if and only if $\rank(B_{\mathcal I}) \le m - r$, where
\[
  B_{\mathcal I} =
  \begin{bmatrix}
    \mathbf{b}_{i_1}^T - \mathbf{b}_{i_0}^T & \mathbf{b}_{i_2}^T - \mathbf{b}_{i_0}^T & \cdots & \mathbf{b}_{i_{\ell - 1}}^T - \mathbf{b}_{i_0}^T
  \end{bmatrix}.
\]
\end{theorem}
\begin{proof}
Suppose that there exists $\mathbf c \in \mathcal C_{\min}$ such that $c_{i} = 1$, $\forall \, i \in \mathcal I$ while  $\rank(B_{\mathcal I}) > m - r$.
Then, from \eqref{Eq:minf}, the following equality must be satisfied:
\[
   \xi_{k} + \sum^{m - 1}_{v = 0} z_{k, v} b_{i, v} = 1, \; \forall \, k \in \{0, 1, \ldots, r - 1\} \text{ and } i \in \mathcal I.
\]
Since $\xi_{k}$ does not depends on $i$,
\[
  \sum^{m - 1}_{v = 0} z_{k, v} b_{i_0, v} = \sum^{m - 1}_{v = 0} z_{k, v} b_{i_1,  v} = \cdots = \sum^{m - 1}_{v = 0} z_{k, v} b_{i_{\ell - 1}, v},
\]
for each $k$, which can be equivalently expressed as
\begin{multline}
  \begin{bmatrix} \mathbf z_0 \\ \mathbf z_1  \\ \vdots \\ \mathbf z_{r - 1} \end{bmatrix}
  \begin{bmatrix}
    \mathbf b_{i_1}^T - \mathbf b_{i_0}^T &  \mathbf b_{i_2}^T - \mathbf b_{i_0}^T & \cdots & \mathbf b_{i_{\ell - 1}}^T - \mathbf b_{i_0}^T
  \end{bmatrix}  \\
   = ZB_{\mathcal I} = \mathbf 0_{r \times \ell},
\end{multline}
where $\mathbf z_k = [z_{k, 0}, z_{k, 1}, \ldots, z_{k, m - 1}]$ for $k \in \{0, 1, \ldots, r - 1\}$. Note that the column space of $B_{\mathcal I}$ lies in the null space of $Z$. Since it is assumed that $\rank(B_{\mathcal I}) > m - r$, the rank-nullity theorem
\[
 \rank (Z) + \rank (B_{\mathcal I}) \le m,
\]
leads to the following inequalities:
\[
  \rank(Z) \le m - \rank(B_{\mathcal I})
                     < m - (m - r) = r.
\]
However, for all $\mathbf c \in \mathcal C_{\min}$, the rank of $Z$ must be $r$, i.e., $\rank(Z) = r$, which contracts the assumption, $\rank(B_{\mathcal I}) > m - r$. Thus, it follows that for a codeword $\mathbf c \in \mathcal C_{\min}$ with $c_i = 1$, $\forall i \in \mathcal I$, the rank of $B_{\mathcal I}$ must be less than or equal to $m - r$, i.e., $\rank(B_{\mathcal I}) \le m - r$.

On the other hand, if $\rank (B_{\mathcal I}) \le m-r$, the rank of the null space of $B_{\mathcal I}^T$ is $ m - \rank(B_{\mathcal I}) \ge r$. Then, we can find a matrix $Z$ whose row space lies in the null space of $B_{\mathcal I}^T$ with $\rank(Z) = r$. Then, the codeword $\mathbf c$ which corresponds to the polynomial
\[
  f(x) = \prod_{k = 0}^{r - 1} \left(1 + \sum_{v = 0}^{m - 1} z_{k, v}b_{i_0, v} + \sum_{v = 0}^{m - 1} z_{k, v}x_v \right)
\]
has $c_i = 1$, $\forall\, i \in \mathcal I$, since
\begin{align}
  f(\mathbf b_{i_j})
    & = \prod_{k = 0}^{r - 1} \left(1 + \sum_{v = 0}^{m - 1} z_{k, v}b_{i_0, v} + \sum_{v = 0}^{m - 1} z_{k, v}b_{i_j, v} \right) \nonumber \\
    & = \prod_{k = 0}^{r - 1} \left(1 + \mathbf z_k \cdot (\mathbf b^T_{i_j} - \mathbf b^T_{i_0}) \right) = 1 \label{Eq:Bmr}
\end{align}
where the equality in \eqref{Eq:Bmr} holds since the row space of the matrix $Z$ lies in the null space of $B_{\mathcal I}^T$. Since $\rank(Z) = r$, the codeword $\mathbf c$ is in  $\mathcal C_{\min}$. Thus, it follows that if $\rank (B_{\mathcal I}^T) \le m-r$, there exists a codeword $\mathbf c \in \mathcal C_{\min}$ with $c_{i} = 1$, $\forall\, i \in \mathcal I$.
\end{proof}

\begin{corollary} \label{Col:punc}
For an index, $p'$ and $\ell < d_{\min}$, if there exists $\mathcal P' \subseteq \mathcal P = \{p_0, p_1, \ldots, p_{\ell - 1}\}$ such that $|\mathcal P'| = \Delta d_{\min} (\mathcal P)$ and $\rank(B_{\mathcal I}) \le m - r$ for $\mathcal I = \mathcal P' \cup \{p'\}$, the minimum distance of punctured RM code with $\mathcal P \cup \{p'\}$ decreases by one from $d_{\min}(\mathcal P)$, i.e., $d_{\min}(\mathcal P \cup \{p'\}) = d_{\min}(\mathcal P) - 1$.  Otherwise, the minimum distance of punctured RM code remains the same, i.e.,  $d_{\min}(\mathcal P \cup \{p'\}) = d_{\min}(\mathcal P)$.
\end{corollary}
\begin{proof}
  Due to Theorem \ref{Thm:Punc}, $c_i = 1$, $\forall \, i \in \mathcal I \subseteq \mathcal P \cup \{p'\}$ if $|\mathcal P'| = \Delta d_{\min} (\mathcal P)$ and $\rank(B_{\mathcal I}) \le m - r$ for $\mathcal I = \mathcal P' \cup \{p'\}$, which leads to  $d_{\min}(\mathcal P \cup \{p'\}) = d_{\min}(\mathcal P) - 1$. Otherwise, $d_{\min}(\mathcal P \cup \{p'\}) = d_{\min}(\mathcal P)$.
\end{proof}

\begin{definition}[Partial Support $\mathcal P'_{\delta}$] \label{Def:PSup}
  The partial support $\mathcal P'_{\delta}$ is defined as a subset of $\mathcal P$ whose cardinality is $\delta$, i.e., $\mathcal P'_{\delta} \subseteq \mathcal P$ and $|\mathcal P'_{\delta}| = \delta$.
\end{definition}

\begin{definition}[Rank Condition for $\delta$ and $p'$] \label{Def:RankC}
The rank condition for $\delta$ says whether there exists a partial support $\mathcal P'_{\delta}$ such that $\rank(B_{\mathcal I}) \le m - r$ for $\mathcal I = \mathcal P'_\delta \cup \{p'\}$.
\end{definition}

The implication of Corollary \ref{Col:punc} can be summarized with the partial support and the rank condition defined in Definitions \ref{Def:PSup} and \ref{Def:RankC}, respectively. That is, Corollary \ref{Col:punc} tells that the updated puncturing pattern with a new index $p'$, i.e., $\mathcal P \cup \{p'\}$ decreases the minimum distance by one, i.e., $d_{\min}(\mathcal P \cup \{p'\}) = d_{\min}(\mathcal P) - 1$ if the rank condition is satisfied for $\delta = \Delta d_{\min}(\mathcal P)$. Thus, to avoid the decrease of minimum distance, a new punctured bit index, $p'$ must be decided as the one for which the rank condition is not satisfied. If there is a single choice of such $p'$, the puncturing pattern must be updated with the index, $p'$, so that the minimum distance remains the same even after puncturing, i.e., $d_{\min}(\mathcal P \cup \{p'\}) = d_{\min}(\mathcal P)$. It fulfills the design goal in \eqref{Eq:minNode1}. Meanwhile, if there are multiple choices of such index $p'$ or none, the proposed puncturing takes into account the codewords of weight $d_{\min}(\mathcal P) + 1$. In particular, the index $p'$ is decided to minimize the number of codewords with weight $d_{\min}(\mathcal P) + 1$ which decreases to $d_{\min}(\mathcal P)$ after puncturing. The selection of such a $p'$ is performed by choosing $p'$ in such a way to minimize the set of partial support, $S_{\Sigma}(p', \Delta d_{\min}(\mathcal P) - 1)$, which is defined in Definition \ref{Def:SP}.

\begin{algorithm}[t]
\DontPrintSemicolon
\small
\caption{Proposed Puncturing Algorithm} \label{Alg:Punc}
\KwInput {$m, r, N_p$}
\KwOutput {Puncturing pattern, $\mathcal P^{(N_p)}$}
\KwInit {$\mathcal P^{(1)} = \{0\}, \Delta d^{(1)} = 1$}
\For {$\ell = 2$ \KwTo $N_p$}{
	$\mathcal B \leftarrow \{0,1, \ldots, N - 1\} \backslash \mathcal P^{(\ell - 1)} $  \;
	$r_{\min}  \leftarrow \min_{p\in \mathcal B} |\mathcal S_{\Sigma}(p, \Delta d^{(\ell - 1)})|$ \;
	$\mathcal B'  \leftarrow \argmin_{p\in \mathcal B} |\mathcal S_{\Sigma}(p, \Delta d^{(\ell - 1)})|$ \;
	\If {$|\mathcal B'| > 1$}{
		$\mathcal B'  \leftarrow \argmin_{p\in \mathcal B'}  |\mathcal S_{\Sigma}(p, \Delta d^{(\ell - 1)} - 1)|$	
	}
	$p' \leftarrow \text{Random}(\mathcal B')$ \tcp*[f]{Randomly select from $\mathcal B'$} \;
	$\mathcal P^{(\ell)}  \leftarrow \mathcal P^{(\ell - 1)}  \cup \{ p'\}$ \;
	\eIf {$r_{\min} > 0$}{
		 $\Delta d^{(\ell)} \leftarrow \Delta d^{(\ell - 1)} + 1 $
	}{
		 $\Delta d^{(\ell)} \leftarrow \Delta d^{(\ell - 1)}$
	}
}
\end{algorithm}

\begin{definition}[Set of Partial Supports $\mathcal S_{\Sigma}(p', \delta)$] \label{Def:SP}
  The set of partial supports $S_{\Sigma}(p', \delta)$ contains all the index sets $\mathcal P'_\delta$ satisfying the rank condition for $\delta$ and $p'$.
\end{definition}

The proposed puncturing scheme based on Corollary \ref{Col:punc} is described in Algorithm \ref{Alg:Punc}. In \cite{macwilliams1977theory}, it is shown that when one coded bit of an RM code is punctured, the punctured RM code has the same weight distribution regardless of the index of the punctured bit. Thus, in Step 1, the algorithm starts with the assumption that the first punctured bit index is zero, i.e., $\mathcal P^{(1)} = \{0\}$. In Step 2, the candidates of the punctured bit are stored in the set denoted by $\mathcal B$. In Steps 3 and 4, the rank condition is tested for the candidates in $\mathcal B$, and a subset $\mathcal B' \subseteq \mathcal B$ is determined to contain the candidates minimizing the set of partial supports, $\mathcal S_{\Sigma}(p', \delta)$. Note that $r_{\min}$ in Step 3 tells the minimum cardinality of the set of partial supports. Thus, $r_{\min} = 0$ indicates that there is no such an index satisfying the rank condition. If there exist multiple candidates in $\mathcal B'$, we update $\mathcal B'$ with the candidates that minimize the number of codewords with the minimum weight, $d_{\min}(P^{(\ell - 1)})$ after puncturing by choosing the candidates minimizing $S_{\Sigma}(p', \Delta d_{\min}(\mathcal P^{(\ell - 1)}) - 1)$ in Step 6. Then, a new punctured bit index, $p'$ is randomly selected from $\mathcal B'$, and the puncturing pattern, $\mathcal P^{(\ell - 1)}$ is now updated with $p'$ in Step 8, which finally results in the puncturing pattern $\mathcal P^{(\ell)}$. In Step 10, the minimum distance reduces by one if the $\ell$-th punctured bit, $p'$ satisfies the rank condition, i.e, $r_{\min} > 0$. Otherwise, the minimum distance remains the same. The design of the puncturing pattern continues until the length of the puncturing pattern reaches $N_p$, i.e. $\ell = N_p$.

\begin{figure}[t]
 \centerline{\includegraphics[width=\figwidth]{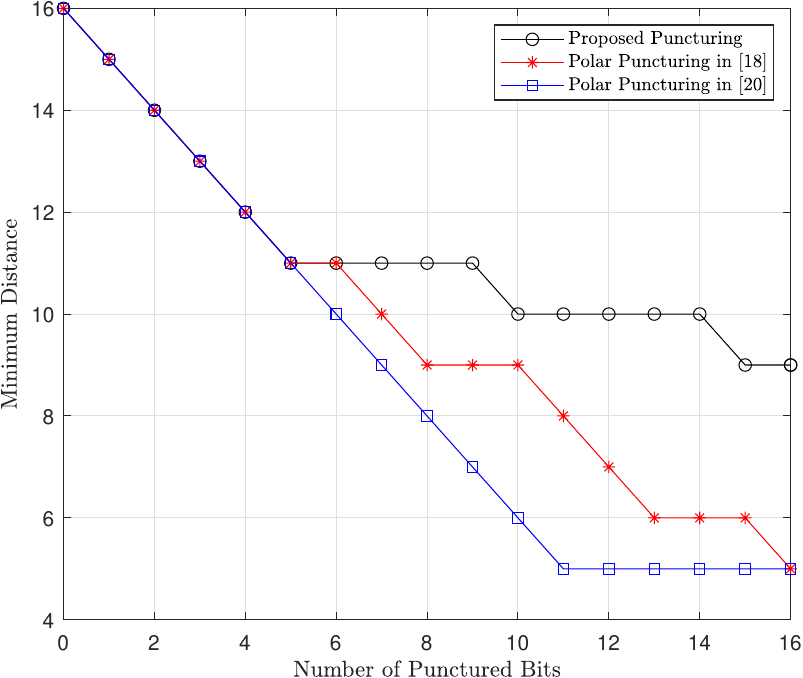}}
	\caption{Minimum distance after puncturing for RM(7, 3) code with the proposed puncturing scheme.} \label{Fig:DminTest}
\end{figure}

In Fig. \ref{Fig:DminTest}, we compare the minimum distance of the punctured polar codes using the proposed puncturing scheme and the existing schemes in \cite{li2019optimal, Han2022rate}. The punctured polar codes are designed by puncturing an RM$(7, 3)$ code whose minimum distance is $d_{\min} = 2^{m-r} = 16$. In Fig. \ref{Fig:DminTest}, it is observed that the scheme in \cite{Han2022rate} provides punctured polar codes whose minimum distance linearly decreases until 11 bits are punctured. Similar behaviors of the minimum distance are witnessed for the punctured polar codes with the scheme in \cite{li2019optimal}. Meanwhile, the comparison in Fig. \ref{Fig:DminTest} clearly demonstrates that the proposed scheme significantly reduces the loss of minimum distance due to the puncturing as compared to the existing schemes \cite{li2019optimal, Han2022rate}.

In the case when the number of punctured bits, $N_p$ is larger than the minimum distance of the base RM code or the base code is not an RM code, the decrease of minimum distance due to the puncturing can be estimated by performing SCLD with a large list size in \cite{li2012adaptive}. While it often requires large computing power and memory requirements, Algorithm \ref{Alg:Punc} can be executed with the numerical estimation of minimum distance in \cite{li2012adaptive}. However, it will be shown that polar codes designed with the proposed scheme have their best performance at $\ell \le d_{\min}$ when an RM code is assumed as the base code. In Section \ref{Sec:Sim}, it will also be demonstrated that the proposed design can be applied to base polar codes other than RM codes.

\section{Reinforcement Learning Based Polar Code Extending} \label{Sec:Ext}
After designing a puncturing pattern, the proposed scheme designs an extending pattern consisting of $N_E$ extended bits. The existing schemes \cite{chen2013hybrid, saber2015incremental} for designing extending patterns mainly focus on improving the SCD performance. However, the SCD-optimized solutions do not necessarily perform well for SCLD of polar codes. In this section, we instead formulate the problem of designing extending patterns with RL and find SCLD-optimized extending patterns for polar codes.

\subsection{Formulation of Designing Extending Patterns as Finite Markov Decision Process}
The problem of designing extending patterns for polar codes can be formulated as an FMDP which is defined with a tuple $\left( \mathcal S,  \mathcal A, P(S'|S, a), r(S, S', a) \right)$. In the definition, $\mathcal S$ and $\mathcal A$ are the state space and the action space, respectively. $P(S'|S, a)$ is the transition probability from a state $S \in \mathcal S$ to a state $S' \in \mathcal S$ with an action $a \in \mathcal A$, and $r(S, S', a)$ is the reward function for the transition from $S$ to  $S'$ with an action $a$. At each time-step $t$, the agent observes the current state $S^{(t)}$ and selects an action $a^{(t)}$ based on the policy $\pi$. The agent interacts with the environment via actions, and the environment provides the reward $r(S^{(t)}, S^{(t + 1)}, a^{(t)})$ to the agent. The next state, $S^{(t+1)}$ is determined by the transition probability $P(S^{(t+1)} | S^{(t)}, a^{(t)})$ given the current state $S^{(t)}$, and the action $a^{(t)}$. The FMDP aims to learn the policy that maximizes the cumulative reward over time.

Now, we formulate the design of extending patterns as a process with a state $S^{(t)}$ which follows the definition of extension matrix defined in Definition \ref{Def:EM} and represented by
\[
  S^{(t)} =
  \begin{bmatrix}
    s_{0, 0}^{(t)} & s_{1, 1}^{(t)} & \cdots & s_{0, m}^{(t)} \\
    s_{1, 0}^{(t)} & s_{1, 1}^{(t)} & \cdots & s_{1, m}^{(t)} \\
     \vdots & \vdots  & \ddots & \vdots  \\
    s_{N - 1,0}^{(t)} & s_{N - 1,1}^{(t)} & \cdots & s_{N - 1,m}^{(t)}
  \end{bmatrix}.
\]
Note that $s_{i, j}^{(t)} = e$ for $j < m$ tells that the bit value at node $(i, j)$ in PCG is repeated $e$ times. Meanwhile, for an unpunctured coded bit, i.e., $(i, j)$ for $i \not \in \mathcal P^{(N_P)}$ and $j = m$, the bit value is repeated and transmitted $e$ times in addition to the coded bit. Thus, the coded bit will be transmitted in total $e + 1$ times. The state is initialized to $S^{(0)} =  \mathbf 0_{N\times (m+1)}$.

An action is defined as the index of a node $a^{(t)} = (i, j) \in \mathcal A$, where $i \in \{0, 1, \ldots, N - 1\}$ and $j \in \{0, 1, \ldots, m\}$. The action indicates that the node with the index of $a^{(t)}$ is repeated and transmitted at time-step $t$. For a given action, the state is updated as follows:
\begin{equation} \label{Eq:State}
s^{(t + 1)}_{i, j} =
  \begin{cases}
    s^{(t)}_{i, j} + 1  & \text{if } a^{(t)} = (i, j), \medskip\\
    s^{(t)}_{i, j}        & \text{otherwise},
  \end{cases}
\end{equation}
for all $0 \le i \le N-1$ and $0 \le j \le m$. The transition between states $S^{(t+1)}$ and $S^{(t)}$ for a given action $a^{(t)}$ is deterministically decided as shown in \eqref{Eq:State}, and thus, the transition probability $P(S^{(t+1)}|S^{(t)}, a^{(t)})$ is either zero or one. At each time-step, the action increases the value of an element in the state by one. Thus, the sum of the state at time-step $t$ amounts to
\[
  \sum_{i, j} s_{i, j}^{(t)} = t.
\]

The state transition is performed until the time-step $t$ reaches $N_E$, i.e., the maximum number of extended bits, and the collection of $N_E$ time-steps is called an \emph{episode}. 

Then, the reward, $r(S^{(t)}, S^{(t + 1)}, a^{(t)})$ is defined to measure the SCLD performance in each episode. To this end, we first transmit random codewords over a physical channel and collect a set of $N_F$ channel outputs, $F^{(0)} = \{\mathbf y_{F, 0}, \mathbf y_{F, 1}, \ldots, \mathbf y_{F, N_F - 1}\}$, that incur SCLD failure. At each time-step $t$, a node $a^{(t)} = (i, j)$ is selected as an action by the agent. The bit value of the node $a^{(t)}$ is repeated and transmitted through the physical channel. The transmission provides extra information to the decoder and thus, the SCL decoder may successfully decode some channel outputs in $F^{(t)}$, and let $D^{(t)} \subseteq F^{(t)}$ be the set of the recovered channel outputs. Note that when the number of channel outputs $N_F$, is sufficiently large, the number of successfully decoded channel outputs can be utilized as a proportional measure of SCLD performance with the extended bit at $a^{(t)}$. Thus, we set the reward as the number of successfully decoded channel outputs $|D^{(t)}|$ after performing the action, i.e., $r(S^{(t)}, S^{(t + 1)}, a^{(t)}) = |D^{(t)}|$. Once the reward is obtained, the set of the channel outputs is updated as $F^{(t+1)} \leftarrow F^{(t)} \backslash D^{(t)}$. 
Since the reward measures the SCLD performance after extending, the policy of action maximizing the reward turns out to be a design rule of extending patterns {that} maximize the SCLD performance.

From \eqref{Eq:State}, we note that the process of designing extending patterns satisfies the Markov property. That is, the transition probability of the state $S^{(t + 1)}$ depends only on the state of the previous time-step, i.e., $S^{(t)}$, 
\begin{equation}\label{Eq:MDP}
  \Pr (S^{(t+1)}|S^{(0)}, S^{(1)}, \ldots, S^{(t)}) = \Pr (S^{(t+1)}|S^{(t)}).
\end{equation}
In addition, since the state is updated for a finite number of time-steps, i.e., $t \le N_E$, and the number of actions is the same as the number of selected nodes in PCG, both the state and action spaces are finite. Thus, the process representing the design of extending patterns is an FMDP, and the design problem, i.e., finding the policy of FMDP, can be efficiently solved using RL techniques, which will be discussed in Section \ref{SubSec:QL}.

\subsection{Q Learning and Deep Q Learning} \label{SubSec:QL}
It is known \cite{watkins1992q} that the optimal policy for an FMDP can be obtained with $Q$-learning which is a mode-free RL algorithm estimating the expected total reward given a pair of state and action. The expected reward value is also called \emph{Q-value}, which is defined as
\begin{equation} \label{Eq:QUpdate}
    Q(S, a) = \sum_{S'} P(S' | S, a) \left(r(S, S', a) + \gamma \max_{a' \in \mathcal A} Q(S', a')\right),
\end{equation}
where $\gamma \in [0,1]$ is the discount factor which indicates the importance of the future reward. Since our model is based on a deterministic transition between states as shown in \eqref{Eq:State}, the $Q$-value is simplified as
\begin{equation}
    Q(S, a) = r(S, S', a) + \gamma  \max_{a' \in \mathcal A} Q(S', a'),
\end{equation}
where $S'$ is the next state for the state $S$ and the action $a$. During the learning process, the $Q$-value is computed as follows:
\begin{multline*}
    Q(S, a) \\ \leftarrow (1-\alpha) Q(S, a) + \alpha \left(r(S, S', a) + \gamma \max_{a' \in \mathcal A} Q(S', a') \right),
\end{multline*}
where $\alpha \in [0, 1]$ is the learning rate. To ensure sufficient exploration, we choose an action based on the $\epsilon$-greedy policy \cite{sutton2018reinforcement} defined as follows:
\[
  \pi_{\epsilon}(S) =
  \begin{cases}
    \argmax_{a \in \mathcal A} Q(S, a) &\text{  with probability  } 1-\epsilon , \\
    \text{random } a \in \mathcal A       &\text{  with probability  } \epsilon,
  \end{cases}
\]
where $\epsilon = \max(\epsilon_{\min}, (1 - \beta)^{t_e})$ for the training-step $t_e$ and the decay rate $\beta$, and $\epsilon_{\min} \in [0, 1]$ is the minimum value of $\epsilon$. The value of $\epsilon$ gradually decreases down to $\epsilon_{\min}$ over training episodes.

In $Q$-learning, the $Q$-values for all state-action pairs can be stored in a look-up table. However, as the dimensions of the state and action spaces grow, the complexity and memory requirements for training and storing the look-up table become prohibitively high. To overcome the problems, we instead utilize deep $Q$-learning (DQL) \cite{mnih2015human} which approximates the $Q$-table with a deep neural network (DNN) denoted by $q_{\boldsymbol \phi}$ with a set of trainable parameters $\boldsymbol \phi$. The DNN $q_{\boldsymbol \phi}$ outputs the approximated $Q$-values, i.e., $q_{\boldsymbol \phi}(S, a)$ for an action $a$ when a state $S$ is given as an input to the DNN. The loss function for training $q_{\boldsymbol \phi}$ is defined as
\[
  \mathcal L(\boldsymbol \phi) = \left(r(S, S', a) + \gamma \max_{a' \in \mathcal A} q_{\phi}(S', a') - q_{\phi}(S, a)  \right)^2.
\]

\begin{figure}[t]
 \centerline{\includegraphics[width=0.99\columnwidth]{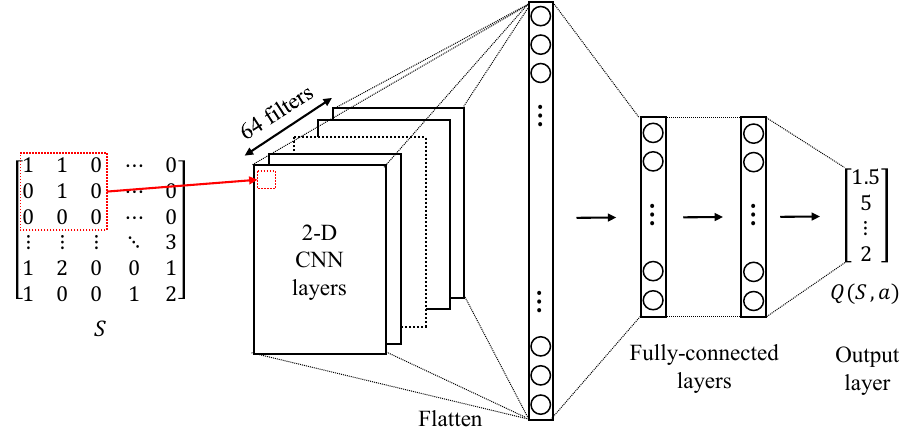}}
	\caption{Estimation of $Q$-values by using CNN for designing extending patterns.} \label{Fig:DQN}
\end{figure}

In Fig. \ref{Fig:PG}, it is noticed that the values of nodes in the PCG are correlated with each other, i.e., spatially correlated, which is well represented with a convolutional neural network (CNN). An implementation of CNN for $Q$-learning is depicted in Fig. \ref{Fig:DQN}, where the state $S$ is the input to the convolutional layer which consists of 64 filters of size $3 \times 3$. The output of the convolutional layers is flattened into a single vector and passed to two fully-connected layers with the rectified linear unit (ReLu) activation function. After the fully-connected layers, the output layer produces the estimated $Q$-values for all possible actions, i.e., $q_{\boldsymbol \phi}(S, a)$, $\forall a \in \mathcal A$.

To improve the stability of the DQL, we use the target network \cite{van2016deep} and replay buffer \cite{lin1992self} for training. In DQL, a single DNN is used for both the selection and evaluation of action, resulting in an overestimation of the $Q$ values \cite{van2016deep}. Such a problem can be mitigated by introducing an independent DNN called target network (denoted by $q_{\boldsymbol \phi'}$) for the action evaluation. Then, the loss function is now re-defined as follows:
 \[
  \mathcal L(\boldsymbol \phi, \boldsymbol \phi') = \left(r(S, S', a) + \gamma \max_{a' \in \mathcal A} q_{\boldsymbol \phi'}(S', a') - q_{\boldsymbol \phi}(S, a)  \right)^2.
 \]
The DNN $q_{\boldsymbol \phi}$ is updated with the back-propagated gradient of $\mathcal L(\boldsymbol \phi, \boldsymbol \phi')$. On the other hand, the parameter for the target network $q_{\boldsymbol \phi'}$ is updated as
\[
   \boldsymbol \phi' \leftarrow (1 - \kappa) \boldsymbol \phi' + \kappa \boldsymbol \phi,
\]
where $\kappa \in [0, 1]$.
Furthermore, we introduce the replay buffer \cite{lin1992self} to store the experiences $(S, a, r, S')$ which are sampled uniformly to update the DNNs. The replay buffer efficiently prevents the dependency between samples, which generates independent samples for training. After sufficient training, the DNN $q_{\boldsymbol \phi}$ becomes capable of closely estimating the $Q$-values. Then, we can find the SCLD-optimized extending pattern $\mathcal E$ from the state ${S^{(N_E)}}$ at the end of an episode where the actions are selected as $a^{(t)} = \argmax_{a' \in \mathcal A} q_{\boldsymbol \phi}(S^{(t)}, a')$ for the state transition during the episode, i.e, $1 \le t \le N_E$.

\subsection{Multi-stage Learning and State Reduction} \label{Subsec:SR}
The number of candidates for extending, i.e., nodes in a PCG, amounts to $\mathcal O (N \log N)$. As the number of extended bits $t$ grows, the size of state space grows exponentially fast, i.e., $|\mathcal S^{(t)}| = {}_{(N\log N + N + t - 1)}C_t$. The exponential growth of the state space not only makes the training complexity unmanageable but also deteriorates the training performance for a finite training set. To resolve the technical issue, this work considers two approaches: 1) multi-stage learning and 2) reduction of state and action spaces. The multi-stage learning reduces the learning complexity by splitting the learning problem into $T$ learning stages. In particular, the multi-stage learning conducts the learning for DNN only for $\ell_T = N_E/T$ (it is assumed that $N_E$ is a multiple of $T$ without loss of generality) extended bits in each learning stage with the state $S^{(t)}_{\lambda}$ where $0 \le t \le \ell_T - 1$, and $0 \le \lambda \le T - 1$ is the learning stage index. At the initial learning stage, the state is initialized with the all-zero matrix, i.e., $S^{(0)}_0 = \mathbf 0_{N\times(m+1)}$ with which the learning for DNN is performed for the first $\ell_T$ extended bits. Then, the learning for the first stage provides a partial set of extended bits, $\mathcal E_0$ containing $\ell_T$ extended bits. After the first learning stage, the state becomes $S^{(\ell_T - 1)}_0$. In the sequel stage, the learning for DNN is fulfilled for the next $\ell_T$ extended bits with $S^{(\ell_T - 1)}_0$ as its initial state, i.e., $S^{(0)}_1 = S^{(\ell_T - 1)}_0$. After performing the learning for stage $\lambda$, the state $S^{(\ell_T - 1)}_{\lambda}$ is obtained, which provides a partial set of extended bits $\mathcal E_\lambda$, where $|\mathcal E_\lambda| = (\lambda + 1)\ell_T$. That is, in each learning stage, $\ell_T$ new extended bits are determined and stored in the partial set of extended bits. The learning for the final learning stage provides $N_E$ extended bits in $\mathcal E_{T - 1}$ which is the designed extending pattern, $\mathcal E$. The multi-stage learning limits the size of state space to ${}_{(N\log N + N + \ell_T - 1)}C_{\ell_T}$, while it goes up to ${}_{(N\log N + N + N_E - 1)}C_{N_E}$ in the conventional learning.

In addition to the multi-stage learning, we propose a novel scheme that greatly reduces the state and action spaces and also the learning complexity. At each time-step $t$, the learning and extended bit selection are performed for a reduced set of nodes, $\mathcal R^{(t)}_{\lambda} \subset \mathcal N$. The proposed scheme makes the reduced set of nodes $\mathcal R^{(t)}_\lambda$ by selecting nodes in PCG that have strong contributions to SCLD performance when they are extended. To this end, it is necessary to have an analytic way to measure such contributions, which unfortunately, to the best of our knowledge, is not available. To turn around the technical obstacle, we instead select nodes by analyzing the reliability propagation from the channel outputs to the information bits through the information trees during SCD, which can be readily performed with the density evolution (DE) technique \cite{mori2009performance, trifonov12efficient}. It is known that the word-error rate (WER) depends on the \emph{weakest} bit-channel reliability. The information bit with the weakest bit-channel reliability will be hereafter called weakest information bit for short. By applying DE, we can identify the weakest information bit and evaluate the reliability values built up at all nodes in PCG after SCD.

The proposed scheme selects nodes in such a way to efficiently improve the reliability of the weakest information bit. In particular, the proposed scheme carefully selects nodes from the descendant nodes in the tree of the weakest information bit since the reliability of each information bit is determined by the reliability propagated through the tree of the information bit. The selection is performed in a greedy manner. Initially, the proposed scheme selects the node for the weakest information bit since the repetition for the weakest information bit directly improves its reliability. Thus, the node for the weakest information bit is the selected node at stage $0$. Then, for stage $j > 0$, a node at stage $j$ is selected among the two child nodes of the parent node, i.e. the selected node at stage $j - 1$, to maximize the reliability of the parent node. The search for the node in $\mathcal R^{(t)}_\lambda$ proceeds until the stage $m$ is reached. Since the weakest information bit may change after a node is extended with $\mathcal R^{(t)}_\lambda$, the set $\mathcal R^{(t + 1)}_\lambda$ has to be updated by conducting DE and performing the search for the nodes in the tree of the new weakest information bit in a greedy manner.

Now, it will be discussed how to determine one of the two child nodes to maximize the reliability of its parent node when the selected node is extended. Note that the reliability of a parent node is determined by the two child nodes with different message indices, say $i$ and $i'$ among which let us assume the index $i$ as the message index of the parent and one child. For example, node (4, 1) in Fig. \ref{Fig:PG} has its two child nodes (4, 2) and (6, 2). In the case that $b_{i, j} = 0$, the LLR value of the parent node $(i, j)$ is given by
\begin{equation} \label{Eq:Boxplus}
    \Lambda_{i, j} = \Lambda_{i, j+1} \boxplus \Lambda_{i', j+1},
\end{equation}
where
\[
  x \boxplus y  = \ln \left (\frac{\cosh \frac{x+y}{2}}{\cosh \frac{x - y}{2} } \right)
	              \approx \sign(x)⋅\sign(y) \min[|x|, |y|],
\]
and $\Lambda_{i,j}$ denotes the LLR value at node $(i, j)$. The reliability is measured as the magnitude of LLR value and expressed as the minimum of the reliabilities for the child nodes, i.e., $\min[|\Lambda_{i, j + 1}|, |\Lambda_{i', j + 1}|]$. Therefore, for the case of $b_{i, j} = 0$, the child node with the smaller reliability is selected to be included in $\mathcal R^{(t)}_\lambda$ since it has a stronger impact on the reliability of the parent node. As discussed, we can acquire the reliability values at all nodes by performing DE, which allows us to identify the one with smaller reliability. Meanwhile, when $b_{i, j} = 1$, e.g., the parent node (3, 1) and its two child nodes (3, 2) and (1, 2), the reliability of the parent node $(i, j)$ is determined as
\begin{equation} \label{Eq:SC}
   \Lambda_{i, j} = (1 - 2 u_{i', j}) \Lambda_{i', j+1} + \Lambda_{i, j+1},
\end{equation}
where $u_{i', j}$ is the hard-decision value of node $(i', j)$. Then, the LLR value of the parent node is the weighted sum of LLR values for the child nodes. Since the hard-decision result $u_{i, j}$ may be erroneous, it is better to choose the child $(i, j + 1)$ as the node at stage $j$. For example, in Fig. \ref{Fig:PG}, the nodes represented by the double circles, i.e., $(4, 0), (4, 1), (6, 2)$, and $(6, 3)$, are selected from the tree of the weakest information bit $(4, 0)$, i.e., $\mathcal R^{(t)}_\lambda = \{(4, 0), (4, 1), (6, 2), (6, 3)\}$.

As discussed in Section \ref{Sec:Polar}, puncturing leads to zero reliability values at some nodes in PCG, which deteriorates the SCLD performance. Thus, when extending patterns are designed with punctured polar codes, the reduced set $\mathcal R^{(t)}_\lambda$ includes the zero set, i.e., $\mathcal N_z$. However, some of the nodes in $\mathcal N_z$ may have already been extended in a previous learning stage, and the the zero set for $\lambda > 0$ must also be updated as follows:
\[
    \mathcal N_{z, \lambda} = \mathcal N_{z, \lambda - 1} \setminus \mathcal E_{\lambda - 1}
\]
where $\mathcal N_{z, \lambda}$ is the zero set at learning stage $\lambda$, and  $\mathcal N_{z, 0} = \mathcal N_z$. Then, the reduced set $\mathcal R^{(t)}_\lambda$ is given by $\mathcal R^{(t)}_\lambda = \mathcal R^{(t)}_\lambda \cup \mathcal N_{z, \lambda}$. Since the learning and design are conducted only for the reduced set $\mathcal R^{(t)}_\lambda$, the size of action space is much reduced to $|R^{(t)}_\lambda| = m + 1 + |\mathcal N_{z, \lambda}|$, which significantly expedites the design process.

\section{Simulation Results}\label{Sec:Sim}
In this section, we demonstrate the efficacy of our proposed design scheme for polar codes by conducting extensive performance evaluations and comparisons. Firstly, we compare the performance of extended polar codes designed with the proposed scheme and an existing scheme in \cite{saber2015incremental}. Then, we carry out performance evaluations of the polar codes designed with the proposed scheme which consists of both puncturing and extending. In addition, the performance of designed polar codes are compared with that of the state-of-the-art polar codes.

For designing extending patterns, we set the parameters for RL as follows: $\kappa = 0.01$, $\beta = 0.005$, $\epsilon_{\min} = 0.01$, and the discount factor $\gamma$ and the experience buffer size are set to 0.99 and $10^4$, respectively. The number of failed samples $N_F$ is set to $100$ for the reward function. The training for the RL is conducted at the starting signal-to-noise ratio (SNR) of the waterfall region (i.e., the region where the WER starts to decrease in the base code) since the decoding performance depends on the decision boundary of the codewords\cite{kim2018communication}. The training is conducted with a batch size of 64 using the adaptive moment estimation (Adam) optimizer \cite{kingma2014adam} with a learning rate $\alpha = 0.01$. For the design parameters, we perform the proposed extending and the SCD-based extending scheme in \cite{saber2015incremental} for $N_E = 10$ to a practical polar code, i.e., a (64, 32) polar code taken from 3GPP NR\cite{3gpp2019nr}. The training for the proposed scheme is conducted with an SCL decoder of list size $L = 8$ at SNR = 1dB.

\begin{figure}[t]
 \centerline{\includegraphics[width=\figwidth]{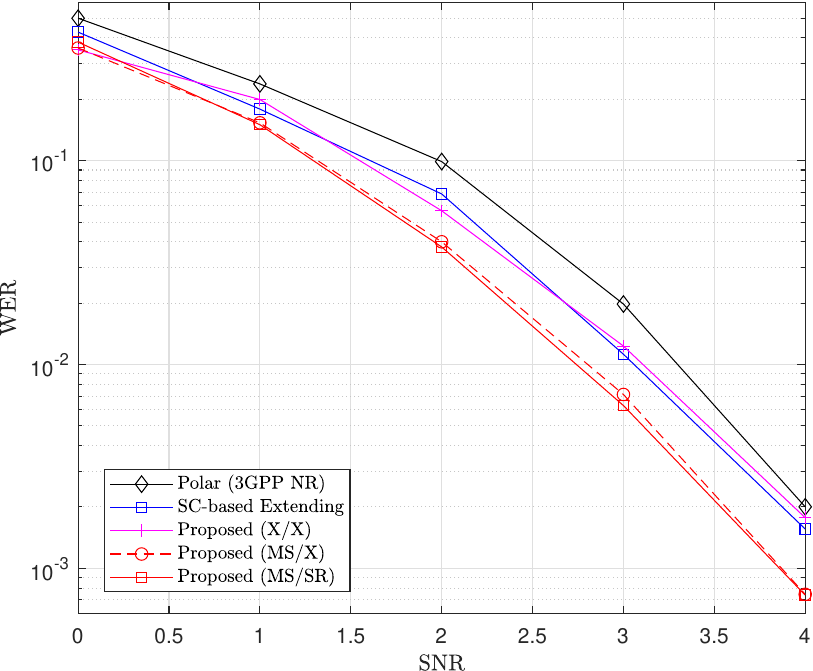}}
	\caption{WER performance comparison of extended polar codes designed with the base code of (64, 32) 3GPP NR polar code decoded with the SCL decoder having $L=8$ and $N_E = 10$. } \label{Fig:PolarExt}
\end{figure}

The performance of the designed extended polar codes is evaluated in terms of WER with the SCL decoder of list size $L = 8$ in Fig. \ref{Fig:PolarExt}, where the proposed scheme without either the multi-stage or state reduction in Section \ref{Subsec:SR} (denoted by X/X) has similar performance to that of SCD-based extending scheme. For short, we will denote the multi-stage and state reduction by MS and SR, respectively. However, the proposed scheme with MS (denoted by MS/X) provides a significant performance gain when $T = 2$ $(\ell_T = 5)$ is assumed to train the RL. In addition, the results in Fig. \ref{Fig:PolarExt} demonstrate that SR allows us to achieve the improved performance at drastically reduced complexity. The extended polar code designed with both MS and SR exhibits approximately 0.35 dB gain as compared to the existing SC-based scheme.

Now, we design polar codes with the proposed design scheme using both puncturing and extending for a (128, 64) RM code as the base code. The design of extended bits is conducted with an SCL decoder of $L = 8$ and MS of $T = 3$ (i.e., $\ell_T = 4$) at SNR = 2dB. To determine the number of transmission holes, the proposed design is performed for different numbers of transmission holes, and the designed polar codes are evaluated in terms of WER at SNR of 2dB in Fig. \ref{Fig:RMDesign}. It is observed that the best performance is obtained at $N_h = 12$ which is less than the minimum distance of the base (128, 64) RM code, i.e., $d_{\min} = 16$. Thus, the puncturing patterns can be designed with Algorithm \ref{Alg:Punc}. In Fig. \ref{Fig:RMDesign}, the WER performance of the designed code with $N_h = 12$ is compared with the (128, 64) polar code in the 3GPP NR standard and the (128, 64) RM-polar code in \cite{li2014rm}. The comparison clearly shows that the proposed polar code outperforms all the other competing polar codes. In particular, the proposed polar code achieves an SNR gain of 0.4dB at a WER of $10^{-2}$ as compared to the 3GPP NR polar code.

\begin{figure}[t]
 \centerline{\includegraphics[width=\figwidth]{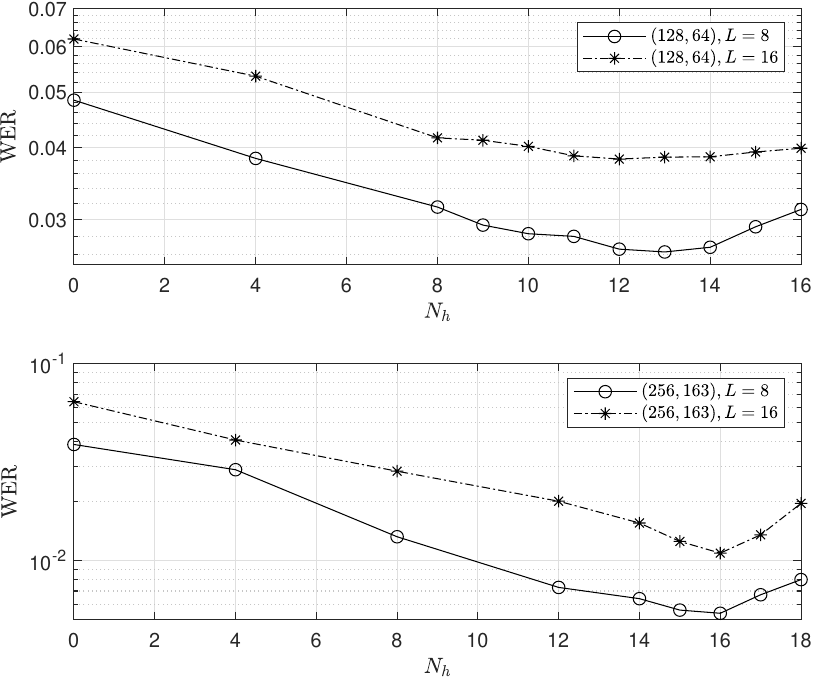}}
	\caption{WER performance of (128, 64) and (256,163) proposed polar codes decoded with the SCL decoder having $L = 8$ and $16$, for different numbers of transmission holes.} \label{Fig:RMDesign3}
\end{figure}

\begin{figure}[t]
 \centerline{\includegraphics[width=\figwidth]{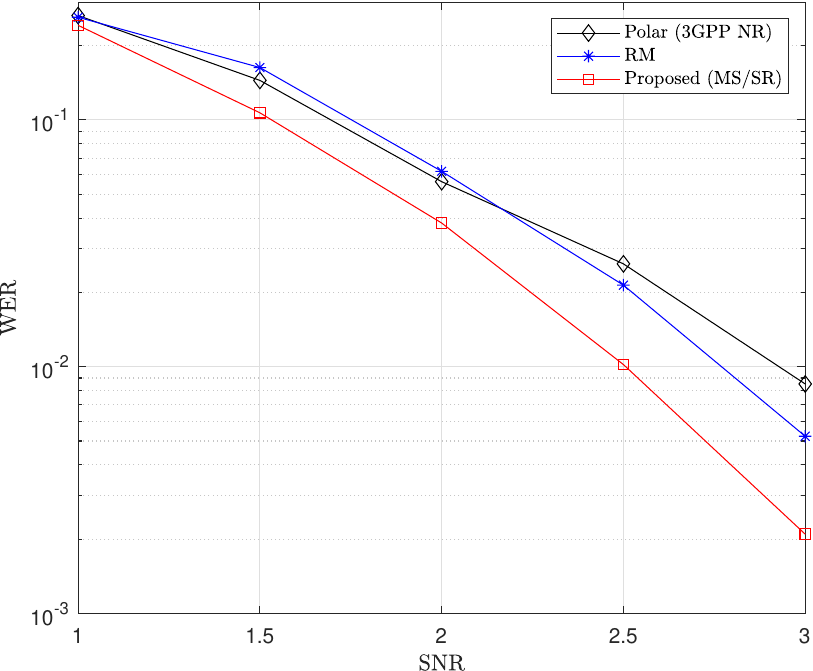}}
	\caption{WER performance comparison of (128, 64) polar codes decoded with the SCL decoder having $L=8$.} \label{Fig:RMDesign}
\end{figure}

\begin{figure}[t]
 \centerline{\includegraphics[width=\figwidth]{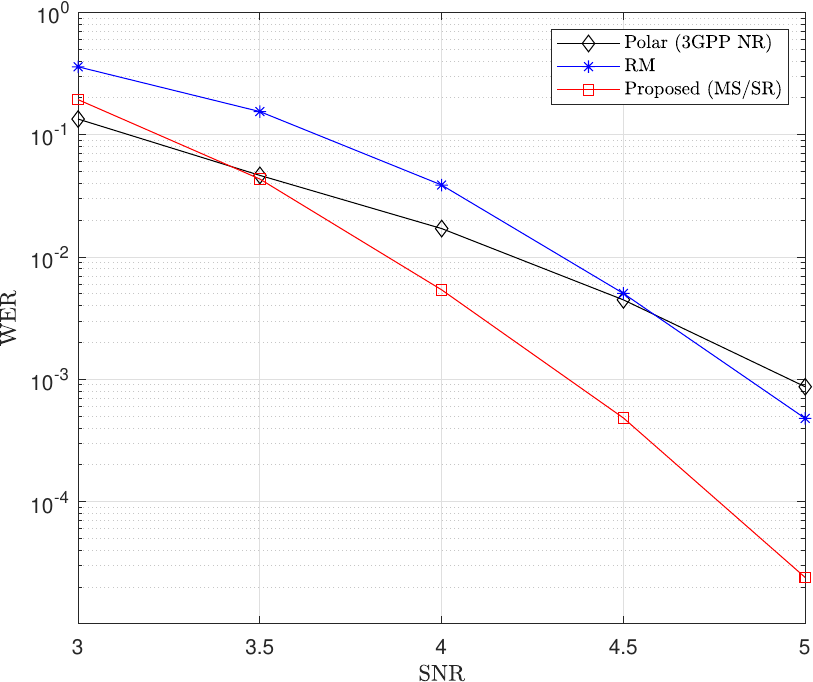}}
	\caption{WER performance comparison of (256,163) polar codes decoded with the SCL decoder having $L = 16$.} \label{Fig:RMDesign3}
\end{figure}

\begin{figure}[t]
 \centerline{\includegraphics[width=\figwidth]{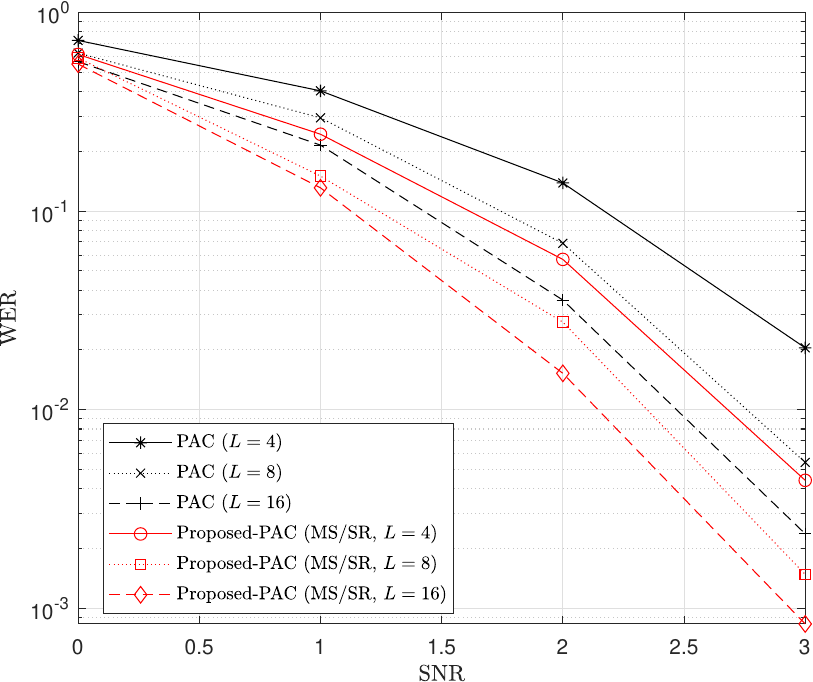}}
	\caption{WER performance comparison of (128,64) PAC codes decoded with the SCL decoder having $L = 4, 8$, and $16$.} \label{Fig:PACDesign}
\end{figure}

The proposed design is applied to a longer code of a higher rate, i.e., a (256, 163) RM code. The number of transmission holes is found by conducting the design for different numbers of transmission holes. The results are shown in Fig. \ref{Fig:RMDesign} where it is noticed that the best performance is obtained at $N_h = 16$. Since the minimum distance of the (256, 163) RM base code is equal to $N_h = 16$, the proposed puncturing algorithm in Algorithm \ref{Alg:Punc} can be utilized for the code design. The code design is performed at SNR of 4dB where the error rate starts to decrease. The performance comparisons in Fig. \ref{Fig:RMDesign3} show that the performance gain with the proposed polar code becomes more distinctive as the code length increases. Compared to the 3GPP NR polar code, the designed code has a coding gain larger than $0.7$dB at a WER of $10^{-3}$.

The proposed scheme utilizes the developed theory of minimum distance when RM codes are assumed as the base code. However, the proposed scheme can be applied to any polar codes by estimating the minimum distance of punctured codes with an SCL decoder \cite{li2012adaptive}. Finally, we will demonstrate how to apply the proposed scheme to a general type of polar codes, e.g., PAC codes \cite{arikan2019sequential}. The design is carried out with a (128, 64) PAC code with an RM profile and the generator polynomial of 0o133 in octal format with constraint length 7 \cite{rowshan2021polarization} as the base code. The change in minimum distance after puncturing the base PAC code is estimated with an SCL decoder of $L = 1024$. Then, the extended bits are designed with $L = 8$, $T = 4$, and $N_h = 12$ ($\ell_T = 3$) at SNR = 2dB. In Fig. \ref{Fig:PACDesign}, the WER performance of the designed code with the PAC code is evaluated with an SCL decoder and compared with the base PAC code. To see the performance variation with different sizes of list, the performance evaluation is conducted with $L = 4$, 8, and 16. Then, it is observed that the proposed polar code has notable performance gains regardless of the list size.

\section{Conclusion}\label{Sec:Con}
In this paper, we propose a novel design scheme for polar codes optimized for the SCLD based on puncturing and extending. In particular, we utilize the minimum distance optimized polar codes as base codes and propose an algorithm that minimizes the decrease of the minimum distance while performing puncturing to generate a degree of freedom on transmission. Since it is intractable to analyze the SCLD performance, we model the problem of designing extending patterns as an FMDP and find SCLD-optimized extending patterns using RL techniques. Moreover, a state reduction scheme is proposed to reduce the training complexity and improve the performance of the trained results. Performance comparisons show considerable error-rate performance improvement for the proposed scheme compared to the polar codes designed with existing schemes. We also show that the proposed scheme improves the performance of modified polar codes, e.g., PAC code. Although this work focuses on performance improvement with the assumption of $N_P = N_E = N_h$, the proposed design provides flexibilities in the code length and rate by considering different numbers of punctured and extended bits, i.e., $N_P \neq N_E$, in addition to the improved performance.

\bibliographystyle{IEEEtran}
\bibliography{DesignPolarPE.bib}
\IEEEtriggeratref{3}

\end{document}